\documentclass[12pt]{article}
\usepackage{color}
\usepackage{bm}
\usepackage{graphicx}
\usepackage{epsfig}
\usepackage{amsmath}
\usepackage{amsfonts}
\usepackage{amssymb}
\usepackage{amsthm}
\usepackage{latexsym}

\newtheorem{theorem}{Theorem}
\newtheorem{lemma}{Lemma}

\usepackage{epsfig}
\usepackage{epstopdf}
\usepackage{pgf,fancyhdr}
\usepackage{float}
\usepackage{amsmath}

\begin{document}


\begin{center}
\LARGE\bf Tightening monogamy and polygamy relations of unified entanglement in multipartite systems
\end{center}

\begin{center}
\rm  Mei-Ming Zhang,$^1$ \  Naihuan Jing,$^{2,1,*}$  \ and  Hui Zhao,$^3$
\end{center}

\begin{center}
\begin{footnotesize} \sl
$^1$ Department of Mathematics, Shanghai University, Shanghai 200444, China

$^2$ Department of Mathematics, North Carolina State University, Raleigh, NC 27695, USA

$^3$ Department of Mathematics, Faculty of Science, Beijing University of Technology, Beijing 100124, China

$^*$ Corresponding author: jing@ncsu.edu

\end{footnotesize}
\end{center}

\begin{center}
\begin{minipage}{15.5cm}
\parindent 20pt\footnotesize
We study the monogamy and polygamy inequalities of unified entanglement in multipartite quantum systems. We first derive the monogamy inequality of unified-$(q, s)$ entanglement for multi-qubit states under arbitrary bipartition, and then obtain the monogamy inequalities of the $\alpha$th ($0\leq\alpha\leq\frac{r}{2}, r\geq\sqrt{2}$) power of entanglement of formation for tripartite states and their generalizations in multi-qubit quantum states. We also generalize the polygamy inequalities of unified-$(q, s)$ entanglement for multi-qubit states under arbitrary bipartition. Moreover, we investigate polygamy inequalities of the $\beta$th ($\beta\geq \max\{1, s\}, 0\leq s\leq s_0, 0\leq s_0\leq\sqrt{2}$) power of the entanglement of formation for $2\otimes2\otimes2$ and $n$-qubit quantum systems. Finally, using detailed examples, we show that the results are tighter than previous studies.
\end{minipage}
\end{center}

\begin{center}
\begin{minipage}{15.5cm}
\begin{minipage}[t]{2.3cm}{\bf Keywords:}\end{minipage}
\begin{minipage}[t]{13.1cm}
Monogamy, Polygamy, Unified-$(q, s)$ entanglement, Entanglement of formation
\end{minipage}\par\vglue8pt

\end{minipage}
\end{center}

\section{Introduction}
Quantum entanglement is an important phenomenon in quantum physics. In multipartite quantum systems, one subsystem's entanglement with other subsystems
 is usually limited to some extent by the entire system.
In other words, the entanglement relation between subsystems cannot be freely and unconditionally transitioned and this property is known as the monogamy of entanglement \cite{1}.
Monogamy relations exist for various entanglement measures and information-theoretic entropies
which underscore their importance and applications in quantum information processing.

Most notable entanglement measures include concurrence, negativity and their generalizations. The first quantitative monogamy relation regarding concurrence was established by Coffman, Kundu and Wootters in three-qubit syatems \cite{2}, and
the CKW inequality was later 
generalized to arbitrary $n$-qubit quantum systems \cite{3}.
Monogamy inequality of negativity, 
like the CKW inequality, was given 
for three-qubit pure states and then extended to multi-qubits \cite{4}.
The authors in \cite{5} derived monogamy relations of the convex-roof extended negativity (CREN) and higher-dimensional extensions.
General monogamy inequalities were provided by the $\alpha$th ($\alpha\geq2$) power of concurrence for multi-qubit states \cite{6}.
A class of monogamy inequalities of $\alpha$th power of CREN regarding multiqubit entanglement for $\alpha\geq1$ were discovered in \cite{7}.
General monogamy relations were also derived for the $\beta$th ($0\leq\beta\leq2$) powers of concurrence, negativity, and CREN in \cite{8}.
The authors in \cite{9} proposed the tighter monogamy relations of the $\alpha$th ($0\leq\alpha\leq2$) power of concurrence under different partition.
Some tighter monogamy inequalities \cite{10, 11, 12}
were obtained for multipartite entangled systems in entanglement distributions.

Further monogamy relations for information-theoretic measures and entropies were discovered for the entanglement of formation (EoF) \cite{13, 14, 15}, the R\'{e}nyi-$q$ entropy \cite{16}, the Tsallis-$q$ entropy \cite{17} and the unified-$(q, s)$ entropy \cite{18}.
Using the Tsallis-$q$ entropy to quantify bipartite entanglement, monogamy of entanglement in multi-qubit systems was proposed in \cite{19}.
The $\alpha$th ($\alpha\geq2$) power of several quantum measures was also found to satisfy certain monogamy inequalities. This type of monogamy relations
was derived for the entanglement of formation ($E^\alpha$) in \cite{20}, the R\'{e}nyi-$q$ entropy ($R_q^\alpha$) in \cite{21}, and the Tsallis-$q$ entropy ($T_q^\alpha$) in \cite{22}.
Moreover, some tight monogamy inequalities of the $\alpha$th-power of unified-$(q, s)$ entanglement for $\alpha\geq1$  were also found for multipartite systems in \cite{23, 24}.
All these monogamy relations were presented separately and derived in different manners, but they displayed some similarity in the format and content. Thus 
 unified and tightened monogamy relations of entanglement measures were studied obtained in \cite{25, 26}. There seems to be a need to formulate a unified
treatment for all these entanglement measures in bipartite systems and even multipartite quantum systems.

 It is known that the assisted entanglement has a dually monogamous property in multipartite systems.  Similarly, polygamy inequalities also provide some bounds for the distribution of entanglement of multipartite quantum states.
The polygamy relation was first established in terms of the entanglement of assistance for three-qubit systems \cite{27}.
It was later generalized to multiqubit systems by using various assisted entanglements \cite{17, 18, 28}. For the arbitrary-dimensional quantum systems, using entanglement of assistance, general polygamy inequalities of multipartite entanglement were also proposed in \cite{29, 30, 31}.
Using Hamming weight of the binary vectors related with the distribution of subsystems, some tighter polygamy inequalities of entanglement of assistance were derived in multipartite quantum systems \cite{24, 32}.
The authors in  \cite{24, 26} provided some polygamy inequalities in terms of unified entanglements.
In \cite{33}, polygamy inequalities of the $\beta$th ($0\leq\beta\leq\alpha$) power of quantum correlations based on residual quantum correlations were presented.


In this paper, we will present a unified and tighter monogamy and polygamy relations for all aforementioned important entanglement measures and entropies, which include the unified-$(q, s)$ entanglement, the R\'{e}nyi-$q$ entropy, the Tsallis-$q$ entropy, and the entanglement of formation for multipartite systems. In other words, our formulation of the
entanglement constraints is done in terms of the unified general $(q, s)$ entropy, which specializes to
the aforementioned various entropies and measurements when $q, s$ take special values. In this way, we hope to see the intrinsic relationships among various monogamy relations. We remark that the $\beta$th ($\beta\geq1$) power of unified entanglement has a different range of $\beta$ from that of \cite{33} and both have no overlaps in multi-qubit quantum systems.
The polygamy inequalities considered in our case for 
 the Tsallis $q$-entropy and $q$-expectation ($q\geq1$) are tighter than those provided in \cite{34}.

The layout of the paper is as follows. In Section 2, we obtain the monogamy inequality of unified-$(q,s)$ entanglement for any multipartite system under arbitrary bipartition. The monogamy inequalities of entanglement of formation for $2\otimes2\otimes2$ and $n$-qubit quantum states are presented. Then the monogamy relation is generalized to several measures of entanglement for multipartite quantum systems. We show that our results are tighter than previous results by detailed examples.
In Section 3, the polygamy inequality of unified entanglement with respect to bipartition is obtained for the multipartite quantum system. Then we derive the polygamy inequalities of entanglement for $2\otimes2\otimes2$ and $n$-qubit quantum systems. We also give examples to show that our bounds are tighter than previous available results. Comments and conclusions are given in Section 4.

\section{Monogamy relations of quantum correlations}\label{sec2}
For a quantum state $\rho$, the unified-$(q,s)$ entropy is defined by \cite{18}:
\begin{eqnarray}
S_{q,s}(\rho):=\frac{1}{(1-q)s}[(\mathrm{tr}\rho^q)^s-1],
\end{eqnarray}
where $q\geq0, q\neq1$ and $s>0$.
The unified-$(q, s)$ entropy specializes to the R\'{e}nyi-$q$ entropy $R_q(\rho)=\frac{1}{1-q}\log[\mathrm{tr}(\rho^q)]$ as $s$ tends to $0$,
the Tsallis-q entropy $T_q(\rho)=\frac{1}{1-q}[\mathrm{tr}(\rho^q)-1]$ as $s$ tends to 1, and
the von Neumann entropy $S(\rho)=-\mathrm{tr}(\rho \log\rho)$ as $q$ tends 1.
For this reason, we also denote $S_{1,s}(\rho)\equiv S(\rho)$ and $S_{q,0}(\rho)\equiv R_q(\rho)$.
The unified-$(q, s)$ entanglement of a bipartite pure state $|\varphi\rangle_{A_1A_2}\in H_{A_1}\otimes\ H_{A_2}$ is defined by
\begin{eqnarray}
E_{q,s}(|\varphi\rangle_{A_1A_2}):=S_{q,s}(\rho_{A_1}),
\end{eqnarray}
where $q,s\geq0$, and $\rho_{A_1}$ is the reduced density matrix of $\rho=|\varphi\rangle_{A_1A_2}\langle\varphi|$, $\rho_{A_1}=\mathrm{tr}_{A_2}(\rho)$. For a mixed bipartite quantum state $\rho_{A_1A_2}=\sum_ip_i|\varphi_i\rangle_{A_1A_2}\langle\varphi_i|\in H_{A_1}\otimes H_{A_2}$, its unified-$(q, s)$ entanglement is defined by the convex roof as usual:
\begin{eqnarray}
E_{q,s}(\rho_{A_1A_2}):=\min\sum_ip_iE_{q,s}(|\varphi_i\rangle_{A_1A_2}),
\end{eqnarray}
where the minimum is taken over all possible convex partitions of $\rho_{A_1A_2}$ into pure state ensembles $\{p_i,|\varphi_i\rangle\}$, $0\leq p_i\leq1$ and $\sum_ip_i=1$.

When $s$ tends to 0 or 1, the unified-$(q, s)$ entanglement of $\rho_{A_1A_2}$ reduces to one-parameter class of entanglement measures---the R\'{e}nyi-$q$ entanglement $R_q(\rho_{A_1A_2})$ or the Tsallis-$q$ entanglement $T_q(\rho_{A_1A_2})$ respectively. As $q$ tends to 1, the unified-$(q, s)$ entanglement of $\rho_{A_1A_2}$ converges to the entanglement of formation (EoF) $E_f(\rho_{A_1A_2})$.

Let $H_{A_1}$ and $H_{A_2}$ be $d_{A_1}$- and $d_{A_2}$-dimensional Hilbert spaces respectively. The concurrence of a bipartite quantum pure state $|\varphi\rangle_{A_1A_2}\in H_{A_1}\otimes\ H_{A_2}$ is defined by \cite{35}:
\begin{eqnarray}
C(|\varphi\rangle_{A_1A_2})=\sqrt{2[1-\mathrm{tr}(\rho_{A_1}^2)]},
\end{eqnarray}
where $\rho_{A_1}$ is the reduced density matrix $\rho_{A_1}=\mathrm{tr}_{A_2}(\rho)$ of $\rho=|\varphi\rangle_{A_1A_2}\langle\varphi|$. For a mixed bipartite quantum state $\rho_{A_1A_2}=\sum_ip_i|\varphi_i\rangle_{A_1A_2}\langle\varphi_i|\in H_{A_1}\otimes H_{A_2}$, the concurrence is given by the convex roof:
\begin{eqnarray}
C(\rho_{A_1A_2})=\min_{\{p_i,|\varphi_i\rangle\}}\sum_ip_iC(|\varphi_i\rangle_{A_1A_2}),
\end{eqnarray}
where the minimum is taken over all possible convex partitions of $\rho_{AB}$ into pure state ensembles $\{p_i,|\varphi_i\rangle\}$, $0\leq p_i\leq1$ and $\sum_ip_i=1$.

The concurrence of a 2-qubit mixed state $\rho$ is given by the remarkable formula \cite{2}:
\begin{eqnarray}
C(\rho)=\max\{\lambda_1-\lambda_2-\lambda_3-\lambda_4, 0\},
\end{eqnarray}
where $\lambda_{i}$, $i=1,\cdots,4$, are the square roots of nonnegative eigenvalues of the matrix $\rho(\sigma_y\otimes\sigma_y)\rho^*(\sigma_y\otimes\sigma_y)$ arranged in decreasing order, $\sigma_y$ is the Pauli matrix, and $\rho^*$ denotes the complex conjugate of $\rho$.

For any $2\otimes d$ pure state $|\varphi\rangle_{A_1A_2}$, the unified-$(q, s)$ entanglement and the concurrence satisfy the functional equation \cite{36}:
\begin{eqnarray}
E_{q, s}(|\varphi\rangle_{A_1A_2})=f_{q,s}(C^2(|\varphi\rangle_{A_1A_2})),
\end{eqnarray}
where $f_{q,s}(x)=\frac{((1+\sqrt{1-x^2})^q+(1-\sqrt{1-x^2})^q)^s-2^{qs}}{(1-q)s2^{qs}}$ with $0\leq x\leq1$. Similar relation holds for 2-qubit mixed states with $0\leq s\leq1$ and $1\leq q\leq\frac{s}{3}$.

For an $n$-qubit quantum state $\rho_{A_1|A_2A_3\cdots A_n}$, the unified-$(q, s)$ entanglement obeys the inequality \cite{23}:
\begin{eqnarray}
E_{q,s}^\alpha(\rho_{A_1|A_2A_3\cdots A_n})\geq E_{q,s}^\alpha(\rho_{A_1A_2})+E_{q,s}^\alpha(\rho_{A_1A_3})+\cdots+E_{q,s}^\alpha(\rho_{A_1A_n}),
\end{eqnarray}
where $\rho_{A_1|A_2A_3\cdots A_n}$ is a quantum state under bipartition ${A_1}$ and ${A_2A_3\cdots A_n}$, $\alpha\geq1$, $q\geq2$, $0\leq s\leq1$, $qs\leq3$.

\begin{lemma}\label{lemma:1}
For real numbers $k\geq1$ and $t\geq k$,

(1) if $0\leq x\leq\frac{1}{2}$, we have
\begin{eqnarray}
(1+t)^x\geq(\frac{1}{2})^x+\frac{(1+k)^x-(\frac{1}{2})^x}{k^x}t^x.
\end{eqnarray}

(2) if $x\geq1$, we have
\begin{eqnarray}
(1+t)^x\leq(\frac{1}{2})^x+\frac{(1+k)^x-(\frac{1}{2})^x}{k^x}t^x.
\end{eqnarray}
\end{lemma}
\begin{proof} Two inequalities are proved similar. We just check the first one.
Consider  $g(x,y)=(1+\frac{1}{y})^{x-1}-(\frac{1}{2})^x$ where $0\leq x\leq\frac{1}{2}$ and $0<y\leq\frac{1}{k}$ with real number $k\geq1$. Then $\frac{\partial g}{\partial x}=(1+\frac{1}{y})^{x-1}ln(1+\frac{1}{y})-(\frac{1}{2})^xln\frac{1}{2}>0$ as $1+\frac{1}{y}\geq2$. So $g(x,y)$ is an increasing function of $x$ when $y$ is fixed, i.e, $g(x,y)\leq g(\frac{1}{2},y)=(1+\frac{1}{y})^{-\frac{1}{2}}-(\frac{1}{2})^{\frac{1}{2}}\leq0$ as $0<(1+\frac{1}{y})^{-1}\leq\frac{1}{2}$. Let $f(x,y)=(1+y)^x-(\frac{1}{2}y)^x$ with $0\leq x\leq\frac{1}{2}$ and $0<y\leq\frac{1}{k}$. As $(1+\frac{1}{y})^{x-1}-(\frac{1}{2})^x\leq0$, then $\frac{\partial f}{\partial y}=xy^{x-1}[(1+\frac{1}{y})^{x-1}-(\frac{1}{2})^x]\leq0$. Thus $f(x,y)$ is a decreasing function of $y$, 
so for $t\geq k$, $f(x,\frac{1}{t})=\frac{(1+t)^x-(\frac{1}{2})^x}{t^x}\geq f(x, \frac1k)=\frac{(1+k)^x-(\frac{1}{2})^x}{k^x}$. Therefore $(1+t)^x\geq(\frac{1}{2})^x+\frac{(1+k)^x-(\frac{1}{2})^x}{k^x}t^x$.
\end{proof}

\begin{lemma}\label{lemma:2}
For nonnegative numbers $p_1\geq p_2\geq\cdots\geq p_n$,

(1) if $0\leq x\leq\frac{1}{2}$, one has
\begin{eqnarray}
(p_1+p_2+\cdots+p_n)^x\geq(\frac{1}{2})^x(l^{n-1}p_1^x+l^{n-2}p_2^x+\cdots+p_n^x),
\end{eqnarray}
where $l=\frac{(1+k)^x-(\frac{1}{2})^x}{k^x}$ with $k\geq1$.

(2) if $x\geq1$, one has
\begin{eqnarray}
(p_1+p_2+\cdots+p_n)^x\leq(\frac{1}{2})^x(l^{n-1}p_1^x+l^{n-2}p_2^x+\cdots+p_n^x),
\end{eqnarray}
where $l=\frac{(1+k)^x-(\frac{1}{2})^x}{k^x}$ with $k\geq1$.
\end{lemma}
\begin{proof} These two inequalities are shown by induction on $n$ similarly. Take the first one for example. The case of $n=1$ holds trivially.
Assume that inequality (11) holds for $n=k$ with $k\geq1$.
Next we consider the case of $n=k+1$. When $p_{k+1}=0$, the inequality (11) holds obviously. Let $p_{k+1}\neq0$ and $\tau=\frac{p_1+p_2+\cdots+p_k}{p_{k+1}}$, we get $\tau\geq k$ as $p_1\geq p_2\geq\cdots\geq p_{k+1}>0$. Then we get
\begin{eqnarray}
(p_1+p_2+\cdots+p_k+p_{k+1})^x&=&p_{k+1}^x(1+\frac{p_1+p_2+\cdots+p_k}{p_{k+1}})^x\nonumber\\
&=&p_{k+1}^x(1+\tau)^x\nonumber\\
&\geq & p_{k+1}^x[(\frac{1}{2})^x+l\tau^x]\nonumber\\
&=&(\frac{1}{2})^xp_{k+1}^x+l(p_1+p_2+\cdots+p_k)^x,
\end{eqnarray}
where $l=\frac{(1+k)^x-(\frac{1}{2})^x}{k^x}$ and the inequality is due to (9) of Lemma 1. Combing this with the inequality of $n=k$ completes the proof. 
\end{proof}

Next we consider the unified entanglement of $\rho_{A_1A_2\cdots A_m|A_{m+1}\cdots A_n}$, $m=1,\cdots,n-1$, with respect to the bipartition $A_1A_2\cdots A_m$ and $A_{m+1}\cdots A_n$. For any $n$-qubit state $\rho_{A_1A_2\cdots A_m|A_{m+1}\cdots A_n}$, $m=1,\cdots, n-1$, on Hilbert space $H_{A_1}\otimes \cdots \otimes H_{A_n}$, the unified entanglement satisfies the following [25]:
\begin{eqnarray}
E_{q,s}^\alpha(\rho_{A_1A_2\cdots A_m|A_{m+1}\cdots A_n})\geq\sum_{i=1}^m\sum_{j=m+1}^nE_{q,s}^\alpha(\rho_{A_iA_j}),
\end{eqnarray}
where $\alpha\geq1$, $q\geq2$, $0\leq s\leq1$, $qs\leq3$, and $\rho_{A_iA_j}$, $i=1,\cdots,m$, $j=m+1,\cdots,n$ are reduced density operators of $\rho_{A_1A_2\cdots A_m|A_{m+1}\cdots A_n}$. Using Lemma 2, we can derive the monogamy inequality of multi-qubit states under arbitrary bipartition based on the $\alpha$th-power of unified-$(q, s)$ entanglement for $0\leq\alpha\leq\frac{r}{2}$ with $r\geq1$.

\begin{theorem}\label{thm:1}
For any $n$-qubit quantum state $\rho_{A_1A_2\cdots A_m|A_{m+1}\cdots A_n}$, $m=1,\cdots, n-1$, and real number $k\geq1$, $q\geq2$, $0\leq s\leq1$ and $qs\leq3$, we have that
\begin{eqnarray}
E_{q,s}^\alpha(\rho_{A_1A_2\cdots A_m|A_{m+1}\cdots A_n})\geq(\frac{1}{2})^{\frac{\alpha}{r}}\sum_{i=1}^{n-m}\sum_{j=0}^{m-1}l^{(m-j)(n-m)-i}E_{q,s}^\alpha(\rho_{A_{j+1}A_{m+i}}),
\end{eqnarray}
where $0\leq\alpha\leq\frac{r}{2}$, $r\geq1$ and $l=\frac{(1+k)^{\frac{\alpha}{r}}-(\frac{1}{2})^{\frac{\alpha}{r}}}{k^{\frac{\alpha}{r}}}$.
\end{theorem}

\begin{proof} We can relabel the subsystems so that $E_{q,s}^r(\rho_{A_iA_j})\geq E_{q,s}^r(\rho_{A_iA_{j+1}})\geq  E_{q,s}^r(\rho_{A_{i+1}A_{m+1}})$ with $i=1,\cdots,m-1$, $j=m+1,\cdots,n-1$.
It follows from
inequality (11) of Lemma 2 and (14) that
\begin{eqnarray}
&&E_{q,s}^\alpha(\rho_{A_1A_2\cdots A_m|A_{m+1}\cdots A_n})\nonumber\\
&=&(E_{q,s}^r(\rho_{A_1A_2\cdots A_m|A_{m+1}\cdots A_n}))^{\frac{\alpha}{r}}\nonumber\\
&\geq&(\sum_{i=1}^m\sum_{j=m+1}^nE_{q,s}^r(\rho_{A_iA_j}))^{\frac{\alpha}{r}}\nonumber\\
&\geq&(\frac{1}{2})^{\frac{\alpha}{r}}(\sum_{i=1}^{n-m}l^{m(n-m)-i}E_{q,s}^\alpha(\rho_{A_1A_{m+i}})+\sum_{i=1}^{n-m}l^{(m-1)(n-m)-i}E_{q,s}^\alpha(\rho_{A_2A_{m+i}})+\cdots\nonumber\\
&&+\sum_{i=1}^{n-m}l^{n-m-i}E_{q,s}^\alpha(\rho_{A_mA_{m+i}}))\nonumber\\
&=&(\frac{1}{2})^{\frac{\alpha}{r}}\sum_{i=1}^{n-m}\sum_{j=0}^{m-1}l^{(m-j)(n-m)-i}E_{q,s}^\alpha(\rho_{A_{j+1}A_{m+i}}),
\end{eqnarray}
where $0\leq\alpha\leq\frac{r}{2}$, $r\geq1$, $k\geq1$, $q\geq2$, $0\leq s\leq1$ and $qs\leq3$.

\end{proof}
{\bf Remark 1.} For the $\alpha$th power of unified-$(q, s)$ entanglement, Theorem 1 provides a general monogamy relation for $0\leq \alpha\leq\frac{r}{2}$ and $r\geq1$. When $s$ tends to $0$ or $1$, Theorem 1 gives
the monogamy inequalities of the R\'{e}nyi-$q$ entanglement or the Tsallis-$q$ entanglement respectively. When $q$ tends to 1, the monogamy inequality for entanglement of formation (EoF) is also obtained from our
general result. 

In the following we discuss the EoF as an analytical unified-$(q, s)$ entanglement under bipartite partition $A_1|A_2A_3\cdots A_n$ and prove tighter monogamy relations. We first give some basic definition. 
Let $H_{A_1}$ and $H_{A_2}$ be $m$ and $n$ ($m\leq n$) dimensional Hilbert spaces respectively. The EoF of a pure quantum state $|\varphi\rangle_{A_1A_2}\in H_{A_1}\otimes H_{A_2}$ is defined by \cite{37}
\begin{eqnarray}
E(|\varphi\rangle_{A_1A_2})=S(\rho_{A_1} ),
\end{eqnarray}
where $\rho_{A_1}=\mathrm{tr}_{A_2}(|\varphi\rangle_{A_1A_2})$ and $S(\rho_{A_1})=-\mathrm{tr}(\rho_{A_1} \log\rho_{A_1})$.
For a mixed bipartite quantum state $\rho_{A_1A_2}=\sum_ip_i|\varphi_i\rangle\langle\varphi_i|\in H_{A_1}\otimes H_{A_2}$, the EoF is given by the convex roof
\begin{eqnarray}
E(\rho_{A_1A_2})=\min_{\{p_i,|\varphi_i\rangle\}}\sum_ip_iE(|\varphi_i\rangle),
\end{eqnarray}
where the minimum is taken over all possible convex partitions of $\rho_{A_1A_2}$ into pure state ensembles $\{p_i,|\varphi_i\rangle\}$, where $0\leq p_i\leq1$ and $\sum_ip_i=1$.

For a $2\otimes m (m\geq2)$ pure state $|\varphi\rangle$, Wootters obtained that $E(|\varphi\rangle)=f(C^2(|\varphi\rangle))$, and $E(\rho)=f(C^2(\rho))$ for 2-qubit mixed states, where $f(x)=h(\frac{1+\sqrt{1-x}}{2})$ and $h(x)=-x\log x-(1-x)\log(1-x)$ in \cite{37}. The function
$f(x)$ is a monotonically increasing one for $0\leq x\leq1$, and $f^{\sqrt{2}}(x^2+y^2)\geq f^{\sqrt{2}}(x^2)+f^{\sqrt{2}}(y^2)$ in \cite{6}. By using $(1+t)^x\geq1+t^x$ for $x\geq1$ and $0\leq t\leq1$, we have $f^r(x^2+y^2)\geq f^r(x^2)+f^r(y^2)$ for $r\geq\sqrt{2}$.
\begin{lemma}\label{lemma:3}
If $f^r(y^2)\geq kf^r(x^2)$, we have
\begin{eqnarray}
f^\alpha(x^2+y^2)\geq(\frac{1}{2})^{\frac{\alpha}{r}}f^\alpha(x^2)+\frac{(1+k)^{\frac{\alpha}{r}}-(\frac{1}{2})^{\frac{\alpha}{r}}}{k^{\frac{\alpha}{r}}}f^\alpha(y^2),
\end{eqnarray}
where $0\leq x,y\leq1$, $0\leq\alpha\leq\frac{r}{2}$, $r\geq\sqrt{2}$, and $k\geq1$.
\end{lemma}
\begin{proof}
When$f^r(y^2)\geq kf^r(x^2)$, we have
\begin{eqnarray}
f^\alpha(x^2+y^2)=f^{r u}(x^2+y^2)&\geq&(f^r(x^2)+f^r(y^2))^u\nonumber\\
&=&f^{r u}(x^2)(1+\frac{f^r(y^2)}{f^r(x^2)})^u\nonumber\\
&\geq& f^{r u}(x^2)[(\frac{1}{2})^u+\frac{(1+k)^u-(\frac{1}{2})^u}{k^u}(\frac{f^r(y^2)}{f^r(x^2)})^u]\nonumber\\
&=&(\frac{1}{2})^uf^{r u}(x^2)+\frac{(1+k)^u-(\frac{1}{2})^u}{k^u}f^{r u}(y^2)\nonumber\\
&=&(\frac{1}{2})^{\frac{\alpha}{r}}f^\alpha(x^2)+\frac{(1+k)^{\frac{\alpha}{r}}-(\frac{1}{2})^{\frac{\alpha}{r}}}{k^{\frac{\alpha}{r}}}f^\alpha(y^2),
\end{eqnarray}
where $0\leq\alpha\leq\frac{r}{2}$ as $0\leq u\leq\frac{1}{2}$, $k\geq1$, the first inequality is obtained by $f^r(x^2+y^2)\geq f^r(x^2)+f^r(y^2)$ for $r\geq\sqrt{2}$ and the second one is due to (9) of Lemma 1.

\end{proof}

For the $n$-qubit quantum state $\rho_{A_1|A_2A_3\cdots A_n}$, regarded as a bipartite state under bipartite partition $A_1|A_2A_3\cdots A_n$, the concurrence satisfies the monogamy inequality for $\alpha\geq2$ \cite{6}:
\begin{eqnarray}
C^\alpha(\rho_{A_1|A_2A_3\cdots A_n})\geq C^\alpha(\rho_{A_1A_2})+C^\alpha(\rho_{A_1A_3})+\cdots+C^\alpha(\rho_{A_1A_n}),
\end{eqnarray}
where $\rho_{A_1A_i}=\mathrm{tr}_{A_2\cdots A_{i-1}A_{i+1}\cdots A_n}(\rho)$, $i=2,\cdots,n$, are the reduced density matrices of $\rho$.

\begin{theorem}\label{thm:2}
For any $2\otimes 2\otimes2$ tripartite state $\rho_{A_1A_2A_3}\in H_{A_1}\otimes H_{A_2}\otimes H_{A_3}$ and real number $k\geq1$,

(1) if $E^r(\rho_{A_1A_3})\geq kE^r(\rho_{A_1A_2})$, then the EoF satisfies
\begin{eqnarray}
E^\alpha(\rho_{A_1|A_2A_3})\geq(\frac{1}{2})^{\frac{\alpha}{r}}E^\alpha(\rho_{A_1A_2})+\frac{(1+k)^{\frac{\alpha}{r}}-(\frac{1}{2})^{\frac{\alpha}{r}}}{k^{\frac{\alpha}{r}}}E^\alpha(\rho_{A_1A_3}),
\end{eqnarray}
where $0\leq\alpha\leq\frac{r}{2}$ and $r\geq\sqrt{2}$.

(2) if $E^r(\rho_{A_1A_2})\geq kE^r(\rho_{A_1A_3})$, then the EoF satisfies
\begin{eqnarray}
E^\alpha(\rho_{A_1|A_2A_3})\geq(\frac{1}{2})^{\frac{\alpha}{r}}E^\alpha(\rho_{A_1A_3})+\frac{(1+k)^{\frac{\alpha}{r}}-(\frac{1}{2})^{\frac{\alpha}{r}}}{k^{\frac{\alpha}{r}}}E^\alpha(\rho_{A_1A_2}),
\end{eqnarray}
where $0\leq\alpha\leq\frac{r}{2}$ and $r\geq\sqrt{2}$.
\end{theorem}
\begin{proof}
Assuming $E^r(\rho_{A_1A_3})\geq kE^r(\rho_{A_1A_2})$, $k\geq1$, we have

\begin{eqnarray}
E^{\alpha}(\rho_{A_1|A_2A_3})&\geq&f^\alpha(C^2(\rho_{A_1|A_2A_3}))\nonumber\\
&\geq&f^\alpha(C^2(\rho_{A_1A_2})+C^2(\rho_{A_1A_3}))\nonumber\\
&\geq&(\frac{1}{2})^{\frac{\alpha}{r}}f^{\alpha}(C^2(\rho_{A_1A_2}))+\frac{(1+k)^{\frac{\alpha}{r}}-(\frac{1}{2})^{\frac{\alpha}{r}}}{k^{\frac{\alpha}{r}}}f^{\alpha}(C^2(\rho_{A_1A_3})) \nonumber\\
&=&(\frac{1}{2})^{\frac{\alpha}{r}}E^{\alpha}(\rho_{A_1A_2})+\frac{(1+k)^{\frac{\alpha}{r}}-(\frac{1}{2})^{\frac{\alpha}{r}}}{k^{\frac{\alpha}{r}}}E^{\alpha}(\rho_{A_1A_3}),
\end{eqnarray}
where $0\leq\alpha\leq\frac{r}{2}$, $r\geq\sqrt{2}$, the first inequality is obtained by $E(\rho_{A_1|A_2A_3})\geq f(C^2(\rho_{A_1|A_2A_3}))$ as qubit states in \cite{10}, the second one is due to inequality (21) and the fact that $f(x)$ is a monotonically increasing function, and the last inequality is due to Lemma 3. The equality holds since $E(\rho)=f(C^2(\rho))$ for 2-qubit states. Similar proof gives  inequality (23) by using Lemma 3.
\end{proof}

For simplicity, denote $E(\rho_{A_1A_i})$, $C(\rho_{A_1A_i})$, $E(\rho_{A_1|A_{j+1}\cdots A_n})$, $C(\rho_{A_1|A_{j+1}\cdots A_n})$ by $E_{A_1A_i}$, $C_{A_1A_i}$, $E_{A_1|A_{j+1}\cdots A_n}$, $C_{A_1|A_{j+1}\cdots A_n}$ respectively, where $i=2,\cdots,n-1$ and $j=1,\cdots,n-1$. Let $l=\frac{(1+k)^{\frac{\alpha}{r}}-(\frac{1}{2})^{\frac{\alpha}{r}}}{k^{\frac{\alpha}{r}}}$ with $0\leq\alpha\leq\frac{r}{2}$, $r\geq\sqrt{2}$ and $k\geq1$. The monogamy inequalities of the $\alpha$th power of the EoF for $n$-qubit quantum states are given by the following theorem for $0\leq\alpha\leq\frac{r}{2}$ and $r\geq\sqrt{2}$.
\begin{theorem}\label{thm:3}
For any $n$-qubit quantum state $\rho_{A_1A_2A_3\cdots A_n}$ and real number $k\geq1$, we have that

(1) if $kE_{A_1A_i}^r\leq E_{A_1|A_{i+1}\cdots A_n}^r$ for $i=2,\cdots,m$ and $E_{A_1A_j}^r\geq kE_{A_1|A_{j+1}\cdots A_n}^r$ for $j=m+1,\cdots,n-1$, $\forall$ $2\leq m\leq n-2$, $n\geq4$, then we have
\begin{eqnarray}
E_{A_1|A_2A_3\cdots A_n}^\alpha&\geq&(\frac{1}{2})^{\frac{\alpha}{r}}(E_{A_1A_2}^\alpha+lE_{A_1A_3}^\alpha+\cdots+l^{m-2}E_{A_1A_m}^\alpha)\nonumber\\
&+&l^m[E_{A_1A_{m+1}}^\alpha+(\frac{1}{2})^{\frac{\alpha}{r}}E_{A_1A_{m+2}}^\alpha+\cdots+(\frac{1}{2})^{\frac{(n-m-2)\alpha}{r}}E_{A_1A_{n-1}}^\alpha]\nonumber\\
&+&l^{m-1}(\frac{1}{2})^{\frac{(n-m-1)\alpha}{r}}E_{A_1A_n}^\alpha,
\end{eqnarray}
where $0\leq\alpha\leq\frac{r}{2}$ and $r\geq\sqrt{2}$.

(2) if $kE_{A_1A_i}^r\leq E_{A_1|A_{i+1}\cdots A_n}^r$ for $i=2,\cdots,n-1$ and $n\geq3$, then we have that
\begin{eqnarray}
E_{A_1|A_2A_3\cdots A_n}^\alpha\geq(\frac{1}{2})^{\frac{\alpha}{r}}(E_{A_1A_2}^\alpha+lE_{A_1A_3}^\alpha+\cdots+l^{n-3}E_{A_1A_{n-1}}^\alpha)+l^{n-2}E_{A_1A_n}^\alpha,
\end{eqnarray}
where $0\leq\alpha\leq\frac{r}{2}$ and $r\geq\sqrt{2}$.

(3) if $E_{A_1A_i}^r\geq kE_{A_1|A_{i+1}\cdots A_n}^r$ for $i=2,\cdots,n-1$ and $n\geq3$, then we have that
\begin{eqnarray}
E_{A_1|A_2A_3\cdots A_n}^\alpha\geq l(E_{A_1A_2}^\alpha+(\frac{1}{2})^{\frac{\alpha}{r}}C_{A_1A_3}^\alpha+\cdots+(\frac{1}{2})^{\frac{(n-3)\alpha}{r}}E_{A_1A_{n-1}}^\alpha)+(\frac{1}{2})^{\frac{(n-2)\alpha}{r}}E_{A_1A_n}^\alpha,
\end{eqnarray}
where $0\leq\alpha\leq\frac{r}{2}$ and $r\geq\sqrt{2}$.
\end{theorem}

\begin{proof}
For arbitrary $2\otimes2\otimes2^{n-2}$ tripartite state, one has in \cite{6}
\begin{eqnarray}
C^2_{A_1|A_2A_3}\geq C^2_{A_1A_2}+C^2_{A_1A_3}.
\end{eqnarray}
For $n$-qubit quantum state $\rho_{A_1A_2A_3\cdots A_n}$, if $kE_{A_1A_i}^r\leq E_{A_1|A_{i+1}\cdots A_n}^r$ for $i=2,\cdots,m$, we have
\begin{eqnarray}
E_{A_1|A_2A_3\cdots A_n}^\alpha&\geq& f^\alpha(C^2_{A_1|A_2A_3\cdots A_n})\nonumber\\
&\geq&f^\alpha(C_{A_1A_2}^2+C_{A_1|A_3\cdots A_n}^2) \nonumber\\
&\geq&(\frac{1}{2})^{\frac{\alpha}{r}}f^\alpha(C_{A_1A_2}^2)+lf^\alpha(C_{A_1|A_3\cdots A_n}^2) \nonumber\\
&\geq&\cdots\nonumber\\
&\geq&(\frac{1}{2})^{\frac{\alpha}{r}}(f^\alpha(C_{A_1A_2}^2)+lf^\alpha(C_{A_1A_3}^2)+\cdots+l^{m-2}f^\alpha(C_{A_1A_m}^2))\nonumber\\
&+&l^{m-1}f^\alpha(C_{A_1|A_{m+1}\cdots A_n}^2)\nonumber\\
&=&(\frac{1}{2})^{\frac{\alpha}{r}}(E_{A_1A_2}^\alpha+lE_{A_1A_3}^\alpha+\cdots+l^{m-2}E_{A_1A_m}^\alpha)+l^{m-1}f^\alpha(C_{A_1|A_{m+1}\cdots A_n}^2),
\end{eqnarray}
where the first inequality follows from $E_{A_1|A_2A_3\cdots A_n}\geq f(C^2_{A_1|A_2A_3\cdots A_n})$ for the $n$-qubit mixed quantum states in \cite{10}, the second one is due to (28) and $f(x)$ being a monotonically increasing function. Using Lemma 3, we get the third inequality. Other inequalities are consequences of Lemma 3
and the last equality holds due to $E(\rho)=f(C^2(\rho))$ for 2-qubit states.

For $E_{A_1A_j}^r\geq kE_{A_1|A_{j+1}\cdots A_n}^r$ for $j=m+1,\cdots,n-1$,  similar proof gives the following inequality by using Lemma 3:
\begin{eqnarray}
f^\alpha(C_{A_1|A_{m+1}+\cdots+A_n}^2)&\geq& lf^\alpha(C_{A_1A_{m+1}}^2)+(\frac{1}{2})^{\frac{\alpha}{r}}f^\alpha(C_{A_1|A_{m+2}\cdots A_n}^2)\nonumber\\
&\geq&\cdots \nonumber\\
&\geq&l[E_{A_1A_{m+1}}^\alpha+(\frac{1}{2})^{\frac{\alpha}{r}}E_{A_1A_{m+2}}^\alpha+\cdots+(\frac{1}{2})^{\frac{(n-m-2)\alpha}{r}}E_{A_1A_{n-1}}^\alpha]\nonumber\\
&+&(\frac{1}{2})^{\frac{(n-m-1)\alpha}{r}}E_{A_1A_n}^\alpha.
\end{eqnarray}
Combining (29) and (30), one obtains (25). If all $kE_{A_1A_i}^r\leq E_{A_1|A_{i+1}\cdots A_n}^r$ for $i=2,\cdots,n-1$ or $E_{A_1A_i}^r\geq kE_{A_1|A_{i+1}\cdots A_n}^r$ for $i=2,\cdots,n-1$, we have the inequality (26) and (27).
\end{proof}

{\bf Remark 2.} Take tripartite quantum states as an example, when $E^r(\rho_{A_1A_3})\geq kE^r(\rho_{A_1A_2})$, the authors in [26] give $E^\alpha(\rho_{A_1|A_2A_3})\geq E^\alpha(\rho_{A_1A_2})+\frac{(1+k)^{\frac{\alpha}{r}}-1}{k^{\frac{\alpha}{r}}}E^\alpha(\rho_{A_1A_3})=\mu_1$. In Theorem 2, the $\alpha$th power of the EoF satisfies $E^\alpha(\rho_{A_1|A_2A_3})\geq(\frac{1}{2})^{\frac{\alpha}{r}}E^\alpha(\rho_{A_1A_2})+\frac{(1+k)^{\frac{\alpha}{r}}-(\frac{1}{2})^{\frac{\alpha}{r}}}{k^{\frac{\alpha}{r}}}E^\alpha(\rho_{A_1A_3})=\mu_2$. Let $\mu=\mu_2-\mu_1$, we find $\mu\geq0$ for $0\leq\alpha\leq\frac{r}{2}$ and $r\geq2$, so our results are tighter than that in \cite{26}.

{\bf Remark 3.} In addition to the EoF, our monogamy relations also work for other quantum correlation measures such as the concurrence by a similar method. In fact, for any $2\otimes 2\otimes2^{n-2}$ tripartite state $\rho_{A_1A_2A_3}\in H_{A_1}\otimes H_{A_2}\otimes H_{A_3}$, $0\leq\alpha\leq\frac{r}{2}$, $r\geq2$, and $k>1$, if $C^r(\rho_{A_1A_3})\geq kC^r(\rho_{A_1A_2})$, the concurrence satisfies
$C^\alpha(\rho_{A_1|A_2A_3})\geq (\frac{1}{2})^{\frac{\alpha}{r}}C^\alpha(\rho_{A_1A_2})+\frac{(1+k)^{\frac{\alpha}{r}}-(\frac{1}{2})^{\frac{\alpha}{r}}}{k^{\frac{\alpha}{r}}}C^\alpha(\rho_{A_1A_3})\geq C^\alpha(\rho_{A_1A_2})+\frac{(1+k)^{\frac{\alpha}{r}}-1}{k^{\frac{\alpha}{r}}}C^\alpha(\rho_{A_1A_3})\geq C^\alpha(\rho_{A_1A_2})+(2^{\frac{\alpha}{r}}-1)C^\alpha(\rho_{A_1A_3})$ since $\frac{(1+k)^{\frac{\alpha}{r}}-1}{k^{\frac{\alpha}{r}}}\geq 2^{\frac{\alpha}{r}}-1$. Thus the conclusion in Theorem 3 is also tighter than that in \cite{8}.

\textit{\textbf{Example 1.}} Consider the quantum state $\rho=|\varphi\rangle\langle\varphi|\in H_1^2\otimes H_2^2\otimes H_3^2$, written in the generalized Schmidt decomposition \cite{38}:
\begin{eqnarray}
|\varphi\rangle=\lambda_0|000\rangle+\lambda_1e^{i\theta}|100\rangle+\lambda_2|101\rangle+\lambda_3|110\rangle+\lambda_4|111\rangle,
\end{eqnarray}
where $0\leq\theta\leq\pi$, $\lambda_i\geq0$, $i=0,\cdots,4$ and $\sum_{i=0}^4\lambda_i^2=1$. We have $C(\rho_{A_1|A_2A_3})=2\lambda_0\sqrt{\lambda_2^2+\lambda_3^2+\lambda_4^2}$, $C(\rho_{A_1A_2})=2\lambda_0\lambda_2$, and $C(\rho_{A_1A_3})=2\lambda_0\lambda_3$. Let $\lambda_0=\lambda_3=\frac{1}{2}$, $\lambda_2=\frac{\sqrt{2}}{2}$, $\lambda_1=\lambda_4=0$, and $k=1.71$, then $E(\rho_{A_1|A_2A_3})=2-\frac{3}{4}\log3\approx0.81$, $E(\rho_{A_1A_2})=-\frac{2+\sqrt{2}}{4}\log{\frac{2+\sqrt{2}}{4}}-\frac{2-\sqrt{2}}{4}\log{\frac{2-\sqrt{2}}{4}}\approx0.60$, $E(\rho_{A_1A_3})=-\frac{2+\sqrt{3}}{4}\log{\frac{2+\sqrt{3}}{4}}-\frac{2-\sqrt{3}}{4}\log{\frac{2-\sqrt{3}}{4}}\approx0.35$. Thus, $E^\alpha(\rho_{A_1|A_2A_3})=0.81^\alpha$. By Theorem 2, the lower bound of $E^\alpha(\rho_{A_1|A_2A_3})$ is $z_1=(\frac{1}{2})^{\frac{\alpha}{r}}0.35^\alpha+\frac{(1+1.71)^{\frac{\alpha}{r}}-(\frac{1}{2})^{\frac{\alpha}{r}}}{1.71^{\frac{\alpha}{r}}}0.6^\alpha$. By Theorem 1 in \cite{26}, the lower bound of $E^\alpha(\rho_{A_1|A_2A_3})$ is $z_2=0.35^\alpha+\frac{(1+1.71)^{\frac{\alpha}{r}}-1}{1.71^{\frac{\alpha}{r}}}0.6^\alpha$. Fig. 1 shows that our result is tighter than that of \cite{26}. To see this clearer, let $z=z_1-z_2=((\frac{1}{2})^{\frac{\alpha}{r}}-1)0.35^\alpha+\frac{1-(\frac{1}{2})^{\frac{\alpha}{r}}}{1.71^{\frac{\alpha}{r}}}0.6^\alpha$. Fig. 2 depicts the value of $z$ for $0\leq\alpha\leq1$ and $r\geq\sqrt{2}$, which confirms that
 Theorem 2 is indeed stronger than that of [26].

\begin{figure}[!htb]
\centerline{\includegraphics[width=0.6\textwidth]{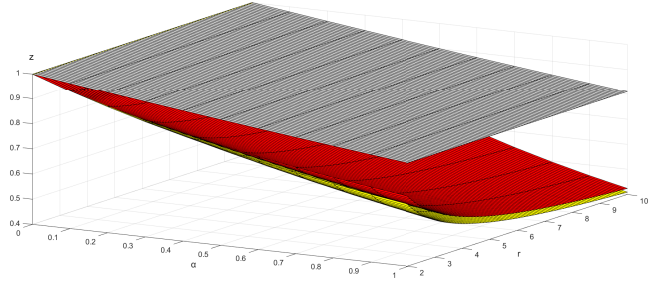}}
\renewcommand{\figurename}{Fig.}
\caption{
The gray surface represents the EoF of the state $|\varphi\rangle$. The lower bound in \cite{26} is shown by the yellow surface and the red surface is our result in Theorem 2.}
\end{figure}

\begin{figure}[!htb]
\centerline{\includegraphics[width=0.6\textwidth]{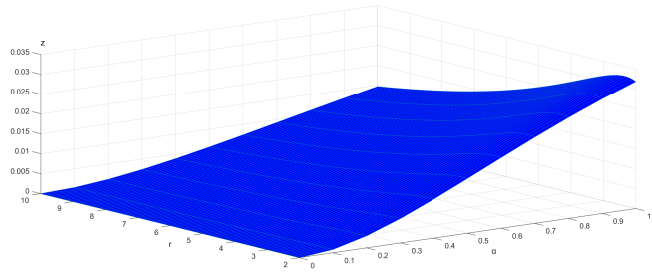}}
\renewcommand{\figurename}{Fig.}
\caption{The blue surface is the difference $z$ between the lower bounds of the entanglement of formation $z_1$ from Theorem 2 and that of in \cite{26}.}
\end{figure}

\section{Polygamy relations of quantum correlations}

In this section, we study the polygamy inequalities for multipartite quantum systems. Recall that the unified-$(q, s)$ entropy of a quantum state $\rho$ satisfies the subadditivity \cite{39}:
\begin{eqnarray}
S_{q,s}(\rho_{A_1A_2})\leq S_{q,s}(\rho_{A_1})+S_{q,s}(\rho_{A_2}),
\end{eqnarray}
where $q>1$, $qs\geq1$. Based on this, we obtain the following result.
\begin{theorem}\label{thm:4}
For any $n$-qubit quantum state $\rho_{A_1A_2\cdots A_m|A_{m+1}\cdots A_n}$ and real number $k\geq1$, suppose that $E_{q,s}(\rho_{A_j|\overline{A_j}})\geq E_{q,s}(\rho_{A_{j+1}|\overline{A_{j+1}}})$, then 
\begin{eqnarray}
E_{q,s}^\beta(\rho_{A_1A_2\cdots A_m|A_{m+1}\cdots A_n})\leq(\frac{1}{2})^\beta(l^{m-1}E_{q,s}(\rho_{A_1|\overline{A_1}})+\cdots+E_{q,s}(\rho_{A_m|\overline{A_m}})),
\end{eqnarray}
where $l=\frac{(1+k)^\beta-(\frac{1}{2})^\beta}{k^\beta}$, $\beta\geq1$, $q>1$, $qs\geq1$, and $\overline{A_i}$, $i=1,\cdots,m$, are the complements of $A_i$ in $\{A_1A_2\cdots A_mA_{m+1}\cdots A_n\}$.
\end{theorem}

\begin{proof}
For a  quantum state $\rho_{A_1A_2\cdots A_m|A_{m+1}\cdots A_n}=\sum_jp_j|\varphi_j\rangle_{A_1A_2\cdots A_m|A_{m+1}\cdots A_n}\langle\varphi_j|$, according to the definition of unified-$(q, s)$ entanglement $E_{q,s}(|\varphi\rangle_{AB}):=S_{q,s}(\rho_A)$ with $q,s\geq0$ from (2) and $E_{q,s}(\rho_{A_j|\overline{A_j}})\geq E_{q,s}(\rho_{A_{j+1}|\overline{A_{j+1}}})$, we have
\begin{eqnarray}
E_{q,s}^\beta(\rho_{A_1A_2\cdots A_m|A_{m+1}\cdots A_n})&=&(\min\sum_jp_jE_{q,s}(|\varphi_j\rangle_{A_1A_2\cdots A_m|A_{m+1}\cdots A_n}))^\beta\nonumber\\
&=&(\min\sum_jp_jS_{q,s}(\rho^j_{A_1A_2\cdots A_m}))^\beta\nonumber\\
&\leq&(\min\sum_j\sum_{i=1}^mp_jS_{q,s}(\rho^j_{A_i}))^\beta\nonumber\\
&=&(\min\sum_j\sum_{i=1}^mp_jE_{q,s}(|\varphi_j\rangle_{A_i|\overline{A_i}}))^\beta\nonumber\\
&=&(\sum_{i=1}^mE_{q,s}(\rho_{A_i|\overline{A_i}}))^\beta\nonumber\\
&\leq&(\frac{1}{2})^\beta(l^{m-1}E_{q,s}(\rho_{A_1|\overline{A_1}})+\cdots+E_{q,s}(\rho_{A_m|\overline{A_m}})),
\end{eqnarray}
where $l=\frac{(1+k)^\beta-(\frac{1}{2})^\beta}{k^\beta}$, $\beta\geq1$, $k\geq1$, $q>1$, $qs\geq1$, $\overline{A_i}$, $i=1,\cdots,m$, are the complements of $A_i$ in $\{A_1A_2\cdots A_mB_1B_2\cdots B_n\}$, the first equality is due to equality (3), the first and second inequality are due to (32) and inequality (12) of Lemma 2 respectively. The first three minima are taken over all possible pure state decompositions of the mixed state $\rho_{A_1A_2\cdots A_m|B_1B_2\cdots B_n}$ while the last minimum is taken over all pure state decompositions of $\rho_{A_i|\overline{A_i}}$.
\end{proof}
{\bf Remark 4.} Theorem 4 provides a general polygamy relation of the unified-$(q, s)$ entanglement for the multipartite quantum system under arbitrary bipartition with $\beta\geq1$. When $s$ tends to $0$ or $1$, the polygamy inequalities of the R\'{e}nyi-$q$ entanglement or the Tsallis-$q$ entanglement are obtained respectively. When $q$ tends to 1, the polygamy inequality for EoF also follows.

In \cite{40}, the authors give the polygamy inequality of entanglement for an n-qubit quantum state $\rho_{A_1A_2A_3\cdots A_n}$, i.e, if there are at least two states such that $C(\rho_{A_1A_{j_1}})C(\rho_{A_1A_{j_2}})\neq0$ for $j_1\neq j_2\in\{2, \cdots, n\}$, then
\begin{eqnarray}
E^s(\rho_{A_1|A_2A_3\cdots A_n})\leq\sum_{i=2}^nE^s(\rho_{A_1A_i}),
\end{eqnarray}
where $0\leq s\leq s_0$, $0<s_0\leq\sqrt{2}$ and $\sum_{i=2}^nE^{s_0}(\rho_{A_1A_i})=1$.
Using inequality (35), one can prove the following Theorem.
\begin{theorem}\label{thm:5} Let $\rho_{A_1A_2A_3}$ be 
a tripartite state in $H_{A_1}\otimes H_{A_2}\otimes H_{A_3}$ and $k\geq1$ a real number.

(1) If $E^s(\rho_{A_1A_3})\geq kE^s(\rho_{A_1A_2})$, then the EoF satisfies
\begin{eqnarray}
E^\beta(\rho_{A_1|A_2A_3})\leq(\frac{1}{2})^{\frac{\beta}{s}}E^\beta(\rho_{A_1A_2})+\frac{(1+k)^{\frac{\beta}{s}}-(\frac{1}{2})^{\frac{\beta}{s}}}{k^{\frac{\beta}{s}}}E^\beta(\rho_{A_1A_3}),
\end{eqnarray}
where $\beta\geq \max\{1, s\}$, $0\leq s\leq s_0$, $0<s_0\leq\sqrt{2}$ and $E^{s_0}(\rho_{A_1A_2})+E^{s_0}(\rho_{A_1A_3})=1$.

(2) If $E^s(\rho_{A_1A_2})\geq kE^s(\rho_{A_1A_3})$, then the EoF satisfies
\begin{eqnarray}
E^\beta(\rho_{A_1|A_2A_3})\leq(\frac{1}{2})^{\frac{\beta}{s}}E^\beta(\rho_{A_1A_3})+\frac{(1+k)^{\frac{\beta}{s}}-(\frac{1}{2})^{\frac{\beta}{s}}}{k^{\frac{\beta}{s}}}E^\beta(\rho_{A_1A_2}),
\end{eqnarray}
where $\beta\geq \max\{1, s\}$, $0\leq s\leq s_0$, $0<s_0\leq\sqrt{2}$ and $E^{s_0}(\rho_{A_1A_2})+E^{s_0}(\rho_{A_1A_3})=1$.
\end{theorem}
\begin{proof}
Assuming $E^s(\rho_{AC})\geq kE^s(\rho_{AB})>0$, we have
\begin{eqnarray}
E^{\beta}(\rho_{A_1|A_2A_3})&\leq&(E^s(\rho_{A_1A_2})+E^s(\rho_{A_1A_3}))^x\nonumber\\
&=&E^{sx}(\rho_{A_1A_2})(1+\frac{E^s(\rho_{A_1A_3})}{E^s(\rho_{A_1A_2})})^x\nonumber\\
&\leq& E^{sx}(\rho_{A_1A_2})[(\frac{1}{2})^x+\frac{(1+k)^x-(\frac{1}{2})^x}{k^x}(\frac{E^s(\rho_{A_1A_3})}{E^s(\rho_{A_1A_2})})^x]\nonumber\\
&=&(\frac{1}{2})^{\frac{\beta}{s}}E^{\beta}(\rho_{A_1A_2})+\frac{(1+k)^{\frac{\beta}{s}}-(\frac{1}{2})^{\frac{\beta}{s}}}{k^{\frac{\beta}{s}}}E^{\beta}(\rho_{A_1A_3}),
\end{eqnarray}
where $\beta\geq \max\{1, s\}$, $0\leq s\leq s_0$, $0<s_0\leq\sqrt{2}$ and $E^{s_0}(\rho_{A_1A_2})+E^{s_0}(\rho_{A_1A_3})=1$, the first inequality is due to (35) and the second one follows from (10) of Lemma 1. Similar argument shows inequality (37) by using Lemma 1.
\end{proof}
Simply denote $E(\rho_{A_1A_i})$ $(i=2,\cdots,n-1)$ by $E_{A_1A_i}$ , $E(\rho_{A_1|A_{j+1}\cdots A_n})$ $(j=1,\cdots,n-1)$ by $E_{A_1|A_{j+1}\cdots A_n}$ and $l=\frac{(1+k)^{\frac{\beta}{s}}-(\frac{1}{2})^{\frac{\beta}{s}}}{k^{\frac{\beta}{s}}}$. Using similar idea of Theorem 5, the polygamy inequality of the $\beta$th power of EoF for an $n$-qubit quantum state is obtained in
the following theorem for $\beta\geq s$, $0\leq s\leq s_0$ and $0<s_0\leq\sqrt{2}$.

\begin{theorem}\label{thm:6}
For any $n$-qubit quantum state $\rho_{A_1A_2A_3\cdots A_n}$ and real number $k\geq1$, we have the following results:

(1) If $kE_{A_1A_i}^s\leq \sum_{j=i+1}^nE^s_{A_1A_j}$ for $i=2,\cdots,m$ and $E_{A_1A_i}^s\geq k\sum_{j=i+1}^nE^s_{A_1A_j}$ for $i=m+1,\cdots,n-1$, $\forall 2\leq m\leq n-2$, $n\geq4$, then
\begin{eqnarray}
E_{A_1|A_2A_3\cdots A_n}^{\beta} &\leq&(\frac{1}{2})^{\frac{\beta}{s}}(E_{A_1A_2}^\beta+lE_{A_1A_3}^\beta+\cdots+l^{m-2}E_{A_1A_m}^\beta)\nonumber\\
&+&l^m[E_{A_1A_{m+1}}^\beta+(\frac{1}{2})^{\frac{\beta}{s}}E_{A_1A_{m+2}}^\beta+\cdots+(\frac{1}{2})^{\frac{(n-m-2)\beta}{s}}E_{A_1A_{n-1}}^\beta]\nonumber\\
&+&l^{m-1}(\frac{1}{2})^{\frac{(n-m-1)\beta}{s}}E_{A_1A_n}^\beta,
\end{eqnarray}
where $\beta\geq s$, $0\leq s\leq s_0$, $0<s_0\leq\sqrt{2}$, and $\sum_{i=2}^nE^{s_0}(\rho_{A_1A_i})=1$.

(2) If $kE_{A_1A_i}^s\leq \sum_{j=i+1}^nE^s_{A_1A_j}$ for $i=2,\cdots,n-1$ and $n\geq3$, then
\begin{eqnarray}
E_{A_1|A_2A_3\cdots A_n}^\beta\leq(\frac{1}{2})^{\frac{\beta}{s}}(E_{A_1A_2}^\beta+lE_{A_1A_3}^\beta+\cdots+l^{n-3}E_{A_1A_{n-1}}^\beta)+l^{n-2}E_{A_1A_n}^\beta,
\end{eqnarray}
where $\beta\geq s$, $0\leq s\leq s_0$, $0<s_0\leq\sqrt{2}$, and $\sum_{i=2}^nE^{s_0}(\rho_{A_1A_i})=1$.

(3) If $E_{A_1A_i}^s\geq k\sum_{j=i+1}^nE^s_{A_1A_j}$ for $i=2,\cdots,n-1$ and $n\geq3$, then
\begin{eqnarray}
E_{A_1|A_2A_3\cdots A_n}^\beta\leq l[E_{A_1A_2}^\beta+(\frac{1}{2})^{\frac{\beta}{s}}E_{A_1A_3}^\beta+\cdots+(\frac{1}{2})^{\frac{(n-3)\beta}{s}}E_{A_1A_{n-1}}^\beta]+(\frac{1}{2})^{\frac{(n-2)\beta}{s}}E_{A_1A_n}^\beta,
\end{eqnarray}
where $\beta\geq s$, $0\leq s\leq s_0$, $0<s_0\leq\sqrt{2}$, and $\sum_{i=2}^nE^{s_0}(\rho_{A_1A_i})=1$.
\end{theorem}

\textit{\textbf{Example 2.}} Consider the W state $|\varphi\rangle=\frac{1}{\sqrt{3}}(|100\rangle+|010\rangle+|001\rangle)$. We have $E(\rho_{A_1|A_2A_3})\approx0.92$, and $E(\rho_{A_1A_2})=E(\rho_{A_1A_3})\approx0.55$. Then $k=1$. Let $E^{s_0}(\rho_{A_1A_2})+E^{s_0}(\rho_{A_1A_3})=1$, then we get $s_0\approx1.16$. We can get $E^\beta(\rho_{A_1|A_2A_3})\leq(\frac{1}{2})^{\frac{\beta}{s}}E^\beta(\rho_{A_1A_2})+\frac{(1+k)^{\frac{\beta}{s}}-(\frac{1}{2})^{\frac{\beta}{s}}}{k^{\frac{\beta}{s}}}E^\beta(\rho_{A_1A_3})=2^{\frac{\beta}{s}}0.55^\beta$ for $0< s\leq1.16$ and $\beta\geq s$. In Fig. 3, we find that our result is tighter when s is larger.
\begin{figure}[!htb]
\centerline{\includegraphics[width=0.6\textwidth]{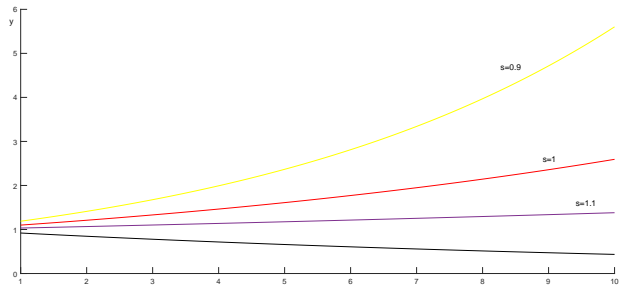}}
\renewcommand{\figurename}{Fig.}
\caption{The y axis is the EoF and its upper bound. The yellow, red, or purple line represents the upper bound from our result for $s=0.9$, $s=1$, or $s=1.1$ respectively and the black line represents the EoF of $|\varphi\rangle$.}
\end{figure}

\section{ Conclusion }
Monogamy and polygamy inequalities of measures for quantum correlation
are one of the fundamental properties for multipartite quantum systems.
In this paper, we have formulated
monogamy inequalities of the unified entanglement for multipartite quantum states under arbitrary bipartition. In particular, we have derived the unified
monogamy inequality of the $\alpha$th ($0\leq\alpha\leq\frac{r}{2}, r\geq\sqrt{2}$) power of the EoF for $2\otimes2\otimes2$ quantum states. Similarly, analytical monogamy inequalities for the $n$-qubit states have been presented. The same method is generalized to the monogamy relations of quantum correlation for multipartite quantum systems. With examples, we have shown that
our results are tighter than the existing ones. Moreover, we have presented for the polygamy inequality of the $\beta$th ($\beta\geq \max\{1, s\}$, $0\leq s\leq s_0$, $0<s_0\leq\sqrt{2}$) power of the EoF for $2\otimes2\otimes2$ quantum states and generalized to the $n$-qubit quantum systems.

\bigskip

\textbf {Acknowledgements}
This work is partially supported by Simons Foundation grant no. 523868
and National Natural Science Foundation of China grant no. 12126351.

\bigskip

\textbf{Data Availability Statement}

All data generated during the study are included in the article.


\end{document}